\documentclass[doublecol,british]{epl2} 
\usepackage[normalem]{ulem}

\usepackage{mathtools}
\usepackage{verbatim}
\usepackage{refstyle}
\usepackage{amsmath}
\usepackage{amsthm}
\usepackage{amssymb}
\usepackage{geometry}
\makeatletter

\AtBeginDocument{}
\AtBeginDocument{}
\AtBeginDocument{}
\AtBeginDocument{\providecommand\eqref[1]{\ref{eq:#1}}}
\AtBeginDocument{}
\AtBeginDocument{}
\AtBeginDocument{}
\AtBeginDocument{\providecommand\defnref[1]{\ref{defn:#1}}}
\AtBeginDocument{\providecommand\tblref[1]{\ref{tbl:#1}}}

\newref{tbl}{name=Table~}{}
\newref{defn}{name=Definition~}{}
\newref{sec}{name=Section~}
\newref{ssec}{name=Subsection~}
\newref{fig}{name=Figure~}

\RS@ifundefined{subref}
  {\def\RSsubtxt{section~}\newref{sub}{name = \RSsubtxt}}
  {}
\RS@ifundefined{thmref}
  {\def\RSthmtxt{Theorem~}\newref{thm}{name = \RSthmtxt}}
  {}
\RS@ifundefined{lemref}
  {\def\RSlemtxt{Lemma~}\newref{lem}{name = \RSlemtxt}}
  {}
\RS@ifundefined{propref}
  {\def\RSproptxt{Proposition~}\newref{prop}{name = \RSproptxt}}
  {}

\theoremstyle{plain}

\theoremstyle{plain}
\newtheorem*{thm*}{\protect\theoremname}

  \theoremstyle{definition}
  \newtheorem{defn}{\protect\definitionname}

  \theoremstyle{remark}
  \newtheorem*{rem}{\protect\remarkname}

  \theoremstyle{remark}
  \newtheorem{defnrem}{\protect\remarkname}[defn]

  \theoremstyle{remark}

  \theoremstyle{plain}
  \newtheorem*{prop*}{\protect\propositionname}
  \theoremstyle{plain}
  
  \theoremstyle{plain}
  
  \theoremstyle{definition}
  
  \theoremstyle{definition}
  \newtheorem*{notn*}{\protect\notationname}

\usepackage{paralist}
\usepackage{graphicx}
\graphicspath{ {images/} }

\usepackage{xcolor}
\usepackage{hyperref}
\hypersetup{
    colorlinks,
    linkcolor={red!50!black},
    citecolor={blue!50!black},
    urlcolor={blue!80!black}
}

\makeatother

\usepackage{babel}
  \providecommand{\definitionname}{Definition}
  \providecommand{\examplename}{Example}
  \providecommand{\lemmaname}{Lemma}
  \providecommand{\propositionname}{Proposition}
  \providecommand{\remarkname}{Remark}
  \providecommand{\theoremname}{Theorem}
  \providecommand{\notationname}{Notation}
\title{Revisiting the admissibility of non-contextual hidden variable models in quantum mechanics}
\shorttitle{Admissibility of non-contextual HV models in QM} 

\author{Atul Singh Arora\inst{1,2} \and Kishor Bharti\inst{3} \and Arvind\inst{1}}
\shortauthor{A. S. Arora \etal}

\institute{                    
  \inst{1} Indian Institute of Science Education and Research (IISER), Mohali,\\Sector 81 SAS Nagar 140306 Punjab India.\\
  \inst{2} Centre for Quantum Information and Communication, Ecole polytechnique de Bruxelles, CP 165,\\Universit\'e libre de Bruxelles (ULB), 1050 Brussels, Belgium \\
  \inst{3} Centre for Quantum Technologies (CQT), \\National University of Singapore (NUS) -  Block S15, 3 Science Drive 2, Singapore 117543.
}


\abstract{ %
We construct a non-contextual
hidden variable model consistent with all { the}
kinematic predictions of quantum mechanics (QM). The famous
Bell-KS theorem shows that non-contextual models which
satisfy a further reasonable restriction are inconsistent
with QM. { In our construction,} we define a
weaker variant of this restriction which captures its
essence while still allowing a non-contextual description of
QM. This is in contrast to { the} contextual
 hidden variable {toy} models, such as
the one by Bell, { and}  brings out an interesting
alternate way of looking at QM. { The results}
also relate to { the Bohmian}
model, where it is harder to pin
down such features.}
\begin{document}

\maketitle


\section{Introduction}
Quantum mechanics (QM) has been one of the most
successful theories in physics so far, however,
there has not yet been a final word on its
completeness and
interpretation~\cite{BellSpkblUnspkbl}.
Einstein's~\cite{EinsteinEPR} work on the
incompleteness of QM and the subsequent seminal
work of Bell~\cite{BellSpkblUnspkbl}, assessing
the compatibility of a more complete model
involving hidden variables (HV) and locality with
QM, has provided deep insights into  how the
quantum world differs from its classical
counterpart.  In recent times, these insights have
been of pragmatic utility in the area of quantum
information processing (QIP), where EPR pairs are
fundamental motifs of
entanglement~\cite{Ekert,PironioRndmnssCrtfcn,NielsenChuang}.
The work of  Kochen Specker (KS)~\cite{KochenSpecker}
and Gleason {\it
et.\,al.}~\cite{Gleason,Peres,Mermin} broadened
the schism between  HV models and QM.  They showed
that it was contextuality and not non-locality
which was at the heart of this schism and the
incompatibility between HV models and QM can arise
even for a single indivisible quantum system.  Contextuality
has thus been identified as a fundamental
non-classical feature of the quantum world and
experiments have also been proposed and conducted
to this effect~\cite{SimonContExpProp,
HuangContExp}.
Contextuality, on the one hand has
led to investigations on the foundational aspects
of QM including attempts to prove the completeness of
QM~\cite{PawelCntxClsscl,CabelloMmryQM}, and on
the other hand has been harnessed for computation
and
cryptography~\cite{HowardCntxCmptn,CabelloCntxScrt}.
{ While there have been generalisations, in this
letter, we restrict ourselves to the standard notion of
non-contextuality as used by KS~\cite{KochenSpecker}.} 

Not all HV models (e.g.
Bohm's model based on
trajectories~\cite{Bohm1,Bohm2}), however, are incompatible
with QM~\cite{BellOnHiddenVariables} and we, in
this letter, present a new non-contextual HV model
consistent with QM.

\begin{figure}[h]
\begin{center}
\includegraphics[width=0.98\columnwidth]{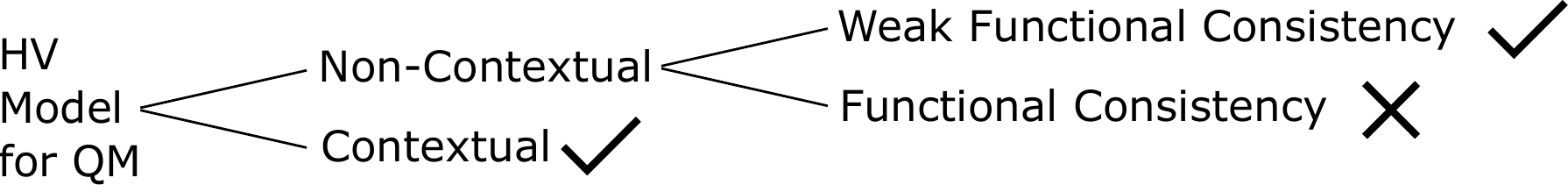}
\end{center}
\caption{Non-contextuality is not inconsistent
with kinematic predictions of QM} \label{fig:block}\end{figure}
{
The KS theorem is applicable to non-contextual HV models
which additionally satisfy \emph{functional-consistency ---
algebraic constraints obeyed by commuting quantum
observables must be satisfied by the HV model at the level
of individual outcomes}. 
One of the main justifications for imposing
this was the requirement that
{ the} observable $\hat A^2$ must depend on $\hat
A$ even at the HV level. This sounds reasonable because
otherwise one can construct HV models where these
observables { can} get mapped to random variables
with no relation to each other (see~\cite{KochenSpecker}). 
As we
will show, one can weaken this requirement into what we call
\emph{weak functional-consistency -- functional-consistency is
demanded only when the observables in question have sharp
values}. This entails that even in the HV model the
observables will depend on each other where they should for
{ consistency} and not otherwise,
thereby capturing the essential idea { without undue
constraints}. The consequence of
demanding only the weaker variant is that the KS theorem no
longer applies and non-contextual HV models can be
consistently constructed. This is in stark contrast with
contextual HV models, such as the toy
model\footnote{introduced in connection with Gleason's
Theorem} by Bell~\cite{BellOnHiddenVariables} because
here we consider models where the algebraic constaints of 
QM are not obeyed in general by the HV model at the level of
individual outcomes. We demonstrate this contrast
by re-examining a `proof of contextuality' using
our model.}
This provides a new
way of looking at the classical-quantum divide and at
the foundations of quantum mechanics.

\section{Non-Contextuality and Functional Consistency}
We introduce some notation and make the relevant notions
precise to facilitate the construction of our
model.
\begin{notn*}
(a) $\psi\in\mathcal{H}$ represents a pure quantum
mechanical state of the system in the Hilbert
space $\mathcal{H}$, (b) $\hat{\mathcal{H}}$ is
defined to mean the set of Hermitian observables
for the system, (c) $[\mathcal{H}]$ is defined to
mean $\{\mathcal{H},\,\mathbb{R}^{\otimes}\}$, which
represents the state  space of the
system including HVs, (d) $[\psi]\in[\mathcal{H}]$
will represent the state of the system including
HVs, (e) a prediction map is $M:\{
\hat{\mathcal{H}},[\mathcal{H}] \}\to\mathbb{R}$,
(f) a sequence map is $S:\{
\hat{\mathcal{H}},[\mathcal{H}],\mathbb{R}
\}\to[\mathcal{H}]$, (g) $f$ is an arbitrary map
from $\{
\hat{\mathcal{H}},\hat{\mathcal{H}},\dots\hat{\mathcal{H}}
\} \to \hat{\mathcal{H}}$ constructed using
multiplication and addition of compatible
observables, and multiplication with complex
numbers, (h) $\tilde{f}$ is a map constructed by
replacing observables in $f$ with real numbers.
\end{notn*}
\begin{defn} A theory is non-contextual, if it
provides a map $M: \{
\hat{\mathcal{H}},[\mathcal{H}] \} \to\mathbb{R}$
to explain measurement outcomes. A theory which is
not non-contextual is contextual~\cite{peresBook}.
\end{defn}

\begin{defnrem}
A prediction map of the form {$M: \{
\hat{\mathcal{H}},[\mathcal{H}] \} \to\mathbb{R}$}
itself can be called  non-contextual.
\end{defnrem}

\begin{defnrem}
{ Broader definitions in the literature have
been suggested~\cite{spekkens_pra} which extend the notion
to probabilistic HV models. 

The idea is that any feature of
a HV model that is not determined solely by the operational
aspect of QM is defined to be a demonstration of contextuality. 
If, for instance, this distinction arises in the measurement procedure
then it is termed as measurement contextuality.
To maintain a
distinction between different features of HV models, which
is of interest here, we stick to the standard definition.}

\end{defnrem}

{ 
In addition to a HV model being non-contextual KS~\cite{KochenSpecker} demand \emph{functional-consistency} to establish their no-go theorem which is defined below in our notation.
}


\begin{defn}
A prediction map $M$ is \emph{functionally-consistent} iff
\begin{align*} &
M(f(\hat{B}_{1},\hat{B}_{2},\dots\hat{B}_{N}),[\psi])
= \\ & \tilde
f(M(\hat{B}_{1},[\psi]),M(\hat{B}_{2},[\psi]),\dots
M(\hat{B}_{N},[\psi])), \end{align*}
where $\hat{B}_{i}\in\hat{\mathcal{H}}$ are
arbitrary mutually commuting observables and
$[\psi]\in[\mathcal{H}]$. A
\emph{non functionally-consistent} map is one that is not
functionally-consistent.  
\end{defn} 
Note that if $M$ is taken to represent the
measurement outcome (in QM), then for states of
the system which are simultaneous eigenkets of
$\hat{B}_{i}$s, $M$ must clearly be
\emph{functionally-consistent}. It is, however, not obvious that
this property must always hold. For example,
consider two spin-half particles in the state 
$\left|1\right\rangle \otimes\left|1\right\rangle
$ written in the computational basis and 
operators 
$\hat{B}_{1}=\hat{\sigma}_{x}\otimes\hat{\sigma}_{x}$,
$\hat{B}_{2}=\hat{\sigma}_{y}\otimes\hat{\sigma}_{y}$
and
$\hat{C}=\hat{B}_{1}\hat{B}_{2}=-\hat{\sigma}_{z}\otimes\hat{\sigma}_{z}$
written in terms of Pauli operators.
We must have $M(\hat{C})=-1$ while
$M(\hat{B}_{1})=\pm1$ and $M(\hat{B}_{2})=\pm1$
independently, according to QM, with probability
half. Here \emph{functional-consistency} clearly is not
required to hold { and indeed this is why it was demanded in addition to being consistent with QM by KS}.
Antithetically, it is clear that if one first
measures $\hat{B}_{1}$ and subsequently measures
$\hat{B}_{2}$, then the product of the results
must be $-1$. This is consistent with measuring
$\hat{C}$. { In our treatment, instead of imposing \emph{functional-consistency}, we demand its aforesaid weaker variant. It captures the same idea, however, only when it has a precise meaning according to QM.} { To this end we define}
\emph{weak functional-consistency} as follows.
\begin{defn} 
A prediction map $m$ has
\emph{weak functional-consistency} for a given
sequence map $s$, iff \begin{align*}
&M(f(\hat{B}_{1},\hat{B}_{2},\dots\hat{B}_{N}),[\psi_{1}])=\\
&\tilde
f(M(\hat{B}_{1},[\psi_{k_{1}}]),M(\hat{B}_{2},[\psi_{k_{2}}]),\dots,M(\hat{B}_{N},[\psi_{k_{N}}])),
\end{align*} 
where
$\{k_{1},k_{2},\dots,k_{N}\} \in\{
N! {\rm ~permutations~of~}  k'{\rm s} \}$, 
$\hat{B}_{i}\in\hat{\mathcal{H}}$ are
arbitrary mutually commuting observables,
$[\psi_{i}]\in[\mathcal{H}]$ and
$[\psi_{k+1}] :=
S(\hat{B}_{k},[\psi_{k}],M(\hat{B}_{k},[\psi_{k}]))$,
$\forall\,[\psi_{i}]$.  
\label{defn:seqnMltpl}
\end{defn} 

With these definitions we are now ready to discuss
the `proof of contextuality'. We first state the
contextuality theorem in our notation:
\begin{thm*} Let a map
$M:\hat{\mathcal{H}}\to\mathbb{R}$, be s.t. (a)
$M(\hat{\mathbb{I}})=1$, (b)
$M(f(\hat{B}_{1},\hat{B}_{2},\dots))=\tilde
f(M(\hat{B}_{1}),M(\hat{B}_{2}),\dots)$, for any
arbitrary function $f$, where $\hat{B}_{i}$ are
mutually commuting Hermitian operators. If $m$ is
assumed to describe the outcomes of measurements,
then no $M$ exists which is consistent with all
predictions of QM. 
\label{thm:KS}
\end{thm*}

\begin{proof} Peres Mermin [PM]
($\left|\mathcal{H}\right|\ge4$)
\cite{Peres,Mermin}:
For a system composed of two spin-half particles 
consider the following set of operators 
\begin{equation*}
\hat{A}_{ij}\doteq\left[\begin{array}{ccc}
\hat{\mathbb{I}}\otimes\hat{\sigma}_{x} &
\hat{\sigma}_{x}\otimes\hat{\mathbb{I}} &
\hat{\sigma}_{x}\otimes\hat{\sigma}_{x}\\
\hat{\sigma}_{y}\otimes\hat{\mathbb{I}} &
\hat{\mathbb{I}}\otimes\hat{\sigma}_{y} &
\hat{\sigma}_{y}\otimes\hat{\sigma}_{y}\\
\hat{\sigma}_{y}\otimes\hat{\sigma}_{x} &
\hat{\sigma}_{x}\otimes\hat{\sigma}_{y} &
\hat{\sigma}_{z}\otimes\hat{\sigma}_{z}
\end{array}\right]
\end{equation*}
which have the property that
all operators along a row (or column) commute. Further,
the product of rows (or columns) yields
$\hat{R}_{i}=\hat{\mathbb{I}}$ and
$\hat{C}_{j}=\hat{\mathbb{I}}\,(j\neq3)$,
$\hat{C}_{3}=-\hat{\mathbb{I}}$, ($\forall\,i,j$) where
$\hat{R}_{i}:=\prod_{j}\hat{A}_{ij}$,
$\hat{C}_{j}:=\prod_{i}\hat{A}_{ij}$. Let us
assume that $M$ exists. From property (b) of the
map, to get $M(\hat{C}_{3})=-1$ (as required by
property (a)), we must have an odd number of $-1$
assignments in the third column. In the remaining
columns, the number of $-1$ assignments must be
even (for each column). Thus, in the entire
square, the number of $-1$ assignments must be
odd. Let us use the same reasoning, but along the
rows. Since each $M(\hat{R}_{i})=1$, we must have
even number of $-1$ assignments along each row.
Thus, in the entire square, the number of $-1$
assignments must be even. We have arrived at a
contradiction and therefore we conclude 
that $M$ does not exist
\end{proof}

\begin{rem} One could in principle assume $M$, to
be s.t. (a) $M(\hat{\mathbb{I}})=1$, (b)
$M(\alpha\hat{B}_{i})=\alpha M(\hat{B}_{i})$, for
$\alpha\in\mathbb{R}$, (c)
$M(\hat{B}_{i}^{2})=M(\hat{B}_{i})^{2}$, (d)
$M(\hat{B}_{i}+\hat{B}_{j})=M(\hat{B}_{i})+M(\hat{B}_{j})$,
to deduce (d)
$M(\hat{B}_{i}\hat{B}_{j})=M(\hat{B}_{i})M(\hat{B}_{j})$
and that $M(\hat{B}_{i})\in$ spectrum of
$\hat{B}_{i}$. Effectively then, condition (b)
listed in the theorem is satisfied as a
consequence.  Therefore, assuming (a)-(d) as
listed above, rules out a larger class of
$M$~\cite{KochenSpecker}.
\end{rem} 
Here $M$ maybe viewed as a specific class of
prediction maps that implicitly depends on the
state $[\psi]$. It is clear that according to the
theorem, non-contextual maps which are
\emph{functionally-consistent} must be incompatible with
QM. 
This leaves open an interesting possibility that non-contextual maps which have \emph{weak functional-consistency} could be consistent with QM.
{  Before proceeding to do so explicitly we observe that \emph{weak functional-consistency} is, in fact, a consequence of QM.}
\begin{prop*} 
Let a quantum mechanical 
system be in a state, s.t. measurement of
$\hat{C}$ yields repeatable results (same result
each time). Then according to QM, \emph{weak functional consistency} holds, where $\hat{C}:=
f(\hat{B}_{1},\hat{B}_{2},\dots\hat{B}_{n})$, and
$\hat{B}_{i}$ are as defined (in \defnref{seqnMltpl})
\end{prop*}
\begin{proof} 
Without loss of generality we can take
$\hat{B}_{1},\hat{B}_{2},\dots\hat{B}_{n}$ to be
mutually compatible and a complete set of
observables (operators can be added to make the set
complete if needed).
It follows
that $\exists$
$\left|\mathbf{b}=\left(b_{1},b_{2},\dots
b_{n}\right)\right\rangle $ s.t.
$\hat{B}_{i}\left|\mathbf{b}\right\rangle
=b_{i}\left|\mathbf{b}\right\rangle $,
and that
$\sum_{\mathbf{b}}\left|\mathbf{b}\right\rangle
\left\langle \mathbf{b}\right|=\hat{\mathbb{I}}$.
Let the state of the system
$\left|\psi\right\rangle $ be s.t.
$\hat{C}\left|\psi\right\rangle
=c\left|\psi\right\rangle $. For
the statement to follow, one need only show that
$\left|\psi\right\rangle $ must be made of only
those $\left|\mathbf{b}\right\rangle $s, which
satisfy $c=\tilde
f(b_{1},b_{2},\dots
b_{n})$.  This is the crucial step and
proving this is straightforward. We start with
$\hat{C}\left|\psi\right\rangle
=c\left|\psi\right\rangle $ and take its inner
product with $\left\langle \mathbf{b}\right|$ to
get 
\begin{eqnarray*} 
c \langle \mathbf{b}|\psi\rangle 
& = &
\langle
\mathbf{b} |\hat{C} |\psi
\rangle  
\\
&=&
\langle
\mathbf{b}| f(\hat{B}_{1},\hat{B}_{2},\dots\hat{B}_{n}) |
\psi \rangle
\\ 
&=&
\tilde f (b_{1},b_{2},\dots b_{n})
\langle
\mathbf{b}|\psi\rangle  
\end{eqnarray*}
Also, we have $\left|\psi\right\rangle
=\sum_{\mathbf{b}}\left\langle
\mathbf{b}|\psi\right\rangle
\left|\mathbf{b}\right\rangle $, from
completeness. If we consider
$\left|\mathbf{b}\right\rangle $s for which
$\left\langle \mathbf{b}|\psi\right\rangle \neq0$,
then we find that indeed $c=\tilde
f(b_{1},b_{2},\dots
b_{n})$. The case when $\left\langle
\mathbf{b}|\psi\right\rangle =0$ is anyway
irrelevant as the corresponding $\vert b\rangle$
does not contribute to $\vert \psi\rangle$.
We can thus conclude that
$\left|\psi\right\rangle $ is made only of those
$\left|\mathbf{b}\right\rangle $s that satisfy the
required relation.
\end{proof} 
It is worth noting that in the Peres Mermin
case, where $\hat{R}_{i}$ and $\hat{C}_{j}$ are
$\pm\hat{\mathbb{I}}$, it follows that all
states are their eigenstates. Consequently, for
these operators \emph{weak functional-consistency}
must always hold.

\section{Construction}
We are now ready to describe our 
model explicitly.
Let the state of a finite dimensional quantum
system be $\left|\chi\right\rangle $. We wish to
assign a value to an arbitrary observable
$\hat{A}=\sum_{a}a\left|a\right\rangle
\left\langle a\right|$, which has eigenvectors $\{
\vert a_j\rangle \}$. The corresponding ordered
eigenvalues are $\{a_j\}$ such that $a_{\rm
min}=a_1$ and $a_{\rm max}=a_n$.  Our HV model for
QM assigns values in the following three steps:
\setdefaultleftmargin{0pt}{}{}{}{}{}
\begin{enumerate}
\item
{\bf  Initial HV:} Pick a number
$c\in[0,1]$, from a uniform random distribution.\
\item
{\bf
Assignment or Prediction:}
 The value assigned to
$\hat{A}$ is given by finding the smallest $a$
s.t.  $c\le\sum_{a'=a_{\text{min}}}^{a}\left|\left\langle
a'|\chi\right\rangle \right|^{2},$ viz. we have
specified a prediction map, $M(\hat{A},[\chi ])=a$.
\item
{\bf  Update:} After measuring an operator, the state is
updated (collapsed) in accordance with the rules
of QM. This completely specifies the sequence map
$S$.
\end{enumerate}
The above HV model works for all quantum systems,
however, to illustrate its working,
consider the example of a spin-half particle in the state
$\left|\chi\right\rangle
=\cos\theta\left|0\right\rangle
+\sin\theta\left|1\right\rangle $ and the observable
$\hat{A}=\hat{\sigma}_{z}=\left|0\right\rangle
\left\langle 0\right|-\left|1\right\rangle
\left\langle 1\right|$. Now, according to the
postulates of this theory, $M(\hat{A},[\chi])=+1$
if
the randomly generated value for 
$c\le\cos^{2}\theta$ else $\hat{A}$ is assigned
$-1$; it then follows,  from $c$ being uniformly
random in $[0,1]$ that the statistics agree with
the predictions of QM {\it i.e.}  the
Born rule. 
The assignment described by the
prediction map $M$ is non-contextual since, given
an operator and a state (alongwith the HV $c$), the value is
uniquely assigned. The map $M$ is, however,
\emph{non functionally-consistent}.

{ To see this \emph{non functional-consistency} explicitly in} our model and to see its
applicability to composite systems,
we apply the model to the Peres Mermin situation of two
spin-half particles. Consider the initial state of
the system 
$\left|\psi_{1}\right\rangle
=\left|00\right\rangle $. Assume we 
obtained $c=0.4$ as a random choice. To arrive
at the assignments, note that
$\left|00\right\rangle $ is an eigenket of only
$\hat{R}_{i},\hat{C}_{j}$ and
$\hat{A}_{33}=\hat{\sigma}_{z}\otimes\hat{\sigma}_{z}$.
Thus, in the first iteration, all these should be
assigned their respective eigenvalues. The
remaining operators must be assigned $-1$ as one
can readily verify by explicitly finding the
smallest $a$ as described in postulate 2 of the
model (see \tblref{HVmodel}).

For the next iteration, $i=2$, 
after the first measurement is over
say the random number generator yielded the value $c=0.1$.
Since $\left|\psi\right\rangle $ is also
unchanged
the assignment remains invariant (in fact any of
$c<0.5$ would yield the same result as
evident from the previous exercise). For the
final step we choose to measure
$\hat{A}_{23}(=\hat{\sigma}_{y}\otimes\hat{\sigma}_{y})$,
to proceed with sequentially measuring
$\hat{C}_{3}$. To simplify calculations, we note
$\left|00\right\rangle
=\frac{1}{\sqrt 2} \left[ \left(\left|\tilde{+}\tilde{-}\right\rangle
+\left|\tilde{-}\tilde{+}\right\rangle
\right)/\sqrt{2}+\left(\left|\tilde{+}\tilde{+}\right\rangle
+\left|\tilde{-}\tilde{-}\right\rangle
\right)/\sqrt{2}\right] , 
$
with
$\left|\tilde{\pm}\right\rangle
=(\left|0\right\rangle \pm i\left|1\right\rangle)/\sqrt{2} $
(eigenkets of $\hat{\sigma}_{y}$). Since
$\left|00\right\rangle$ is manifestly not an
eigenket of $\hat A_{23}$, we must find an appropriate
eigenket $\left|a^-_{23}\right\rangle $ s.t.
$\hat A_{23}\left|a^{-}_{23}\right\rangle
=-\left|a^{-}_{23}\right\rangle $, since $c=0.1$ and
$\left\langle a_{23}^{-}|00\right\rangle $ is already
$>0.1$. It is evident that
$\left|a^{-}_{23}\right\rangle
=\left(\left|\tilde{+}\tilde{-}\right\rangle
+\left|\tilde{-}\tilde{+}\right\rangle
\right)/\sqrt{2}=\left(\left|00\right\rangle
+\left|11\right\rangle \right)/\sqrt{2}$, which
becomes the final state.

For the final
iteration, $i=3$, say we obtain $c=0.7$. So far,
we have $M_{1}(\hat{A}_{33})=1$ and
$M_{2}(\hat{A}_{23})=-1$. We must obtain
$M_{3}(\hat{A}_{13})=1$, independent of the value
of $c$, to be consistent. Let us check that.
Indeed, according to postulate 2, since
$\hat{\sigma}_{x}\otimes\hat{\sigma}_{x}\left(\left|00\right\rangle
+\left|11\right\rangle
\right)/\sqrt{2}=1\left(\left|00\right\rangle
+\left|11\right\rangle \right)/\sqrt{2}$,
$M_{3}(\hat{A}_{13})=1$ for all allowed values of
$c$. It is to be noted that
$M_{2}(\hat{A}_{33})=M_{3}(\hat{A}_{33})$ and
$M_{2}(\hat{A}_{23})=M_{3}(\hat{A}_{23})$, which
essentially expresses the compatibility of these
observables; i.e. 
once measured, the values of observables compatible
with $\hat{A}_{13}$ are not affected by the 
measurement of $\hat{A}_{13}$.

\begin{table*}
 \begin{equation*}
\renewcommand{\arraystretch}{1.5}
\begin{array}{ccc} 
i=1:\,\, c=0.4,\,\left|\psi_{\text{init}}\right\rangle=\left|00\right\rangle
&
i=2:\,\, c=0.1,\,\left|\psi_{\text{init}}\right\rangle=\left|00\right\rangle
&
i=3:\,\, c=0.7,\,\left|\psi_{\text{init}}\right\rangle=\frac{1}{\sqrt{2}}
(\left|00\right\rangle
+\left|11\right\rangle) \\
\renewcommand{\arraystretch}{1}
M_{1}(\hat{A}_{ij})\doteq\left[\begin{array}{ccc} -1 & -1 & -1\\ -1 & -1 & -1\\ -1 & -1 & +1\end{array}\right] &
\renewcommand{\arraystretch}{1}
M_{2}(\hat{A}_{ij})\doteq\left[\begin{array}{ccc}-1 & -1 & -1\\ -1 & -1 & -1\\ -1 & -1 & +1\end{array}\right] &
\renewcommand{\arraystretch}{1}
M_{3}(\hat{A}_{ij})\doteq\left[\begin{array}{ccc}+1
& +1 & +1\\ +1 & +1 & -1\\ +1 & +1 &
+1\end{array}\right] \\
M_{1}(\hat{R}_{i}), M_{1}(\hat{C}_{j})=+1\,(j\neq3) &
M_{2}(\hat{R}_{i}), M_{2}(\hat{C}_{j})=+1\,(j\neq3) &
M_{3}(\hat{R}_{i}), M_{3}(\hat{C}_{j})=+1\,(j\neq3)\\
M_{1}(\hat{C}_{3})=-1 & M_{2}(\hat{C}_{3})=-1 & M_{3}(\hat{C}_{3})=-1\\ 
\hat{A}_{33}=\hat{\sigma}_{z}\otimes\hat{\sigma}_{z};M_{1}(\hat{A}_{33})=+1&
\quad\quad\hat{A}_{23}=\hat{\sigma}_{y}\otimes\hat{\sigma}_{y};M_{2}(\hat{A}_{23})=-1\quad\quad&
\hat{A}_{13}=\hat{\sigma}_{x}\otimes\hat{\sigma}_{x};M_{3}(\hat{A}_{13})=+1\\
\left|\psi_{\text{final}}\right\rangle =
\left|00\right\rangle  &
\left|\psi_{\text{final}}\right\rangle =
\frac{1}{\sqrt{2}}(\left|00\right\rangle
+\left|11\right\rangle) &
\left|\psi_{\text{final}}\right\rangle =
\frac{1}{\sqrt{2}}(\left|00\right\rangle+\left|11\right\rangle)


\label{eq:toyModel}
\end{array}
\end{equation*}
\caption{HV model applied to the Peres Mermin situation}
\label{tbl:HVmodel}
\end{table*}
The \emph{non functional-consistency} is
manifest, for $M_{1}(\hat{C}_{3})=1\neq
M_{1}(\hat{A}_{13})M_{1}(\hat{A}_{23})M_{1}(\hat{A}_{33})=-1$,
where the subscript refers to the iteration
number.  More precisely,
$M_{1}(\hat{O}):=M(\hat{O},\left[\left|\psi_{1}\right\rangle
=\left|00\right\rangle \right])$ where the
complete state $\left[\left|\psi_{1}\right\rangle
\right]$ implicitly refers to both the quantum
state $\left|00\right\rangle $ and the HV $c=0.4$.
The model, however, obeys the \emph{weak functional-consistency} requirement,
namely,
$M_{1}(\hat{C}_{3})=M_{1}
(\hat{A}_{33})M_{2}(\hat{A}_{23})M_{3}(\hat{A}_{13})$,
where
$M_{2}:=M(\hat{O},\left[\left|\psi_{2}\right\rangle
\right])$,
$M_{3}:=M(\hat{O},\left[\left|\psi_{3}\right\rangle
\right])$ and $\left|\psi_{2}\right\rangle
,\,\left|\psi_{3}\right\rangle $ are obtained from
postulate 3. Note that for each iteration, a new
HV is generated. 

\section{Implications and Remarks}
{  
The model demonstrates \emph{non function-consistency} as an alternative signature of quantumness as opposed to contextuality. This view is not just an artifact of the simplicity of the model. 
It }has implications to Bohm's HV
model (BHVM), where if the initial conditions are
precisely known the entire trajectory of a
particle (guided by the wavefunction) can be
predicted including the individual outcome of
measurements.  BHVM applied to a single spin-half
particle in a Stern-Gerlach experiment, can
predict opposite results for two measurements of
the same operator.  This observation which is
often used as a demonstration of ``contextuality''
in BM, does not involve  overlapping sets of
compatible measurements to provide two different
contexts~\cite{HardyCntxBM} {and is therefore not of
interest here.}  It turns out that,
if one constructs  a one-one map between an
experiment and an observable (by following a
certain convention), this so-called ``contextuality''
can be removed from BHVM.  However, the prediction
map so obtained from observables to measurement
outcomes turns out to be \emph{non functionally-consistent},
suggesting that \emph{non functional-consistency} is a more
suitable explanation of non-classicality
of BHVM.
{ In fact, our model when appropriately extended
to continuous variables yields Bohmian trajectories in the
single particle  case which suggests that it can be used as
a starting point for constructing more interesting families
of HV theories that take quantum time dynamics into account
as well.

Bell himself had constructed a deceptively similar toy
model\footnote{as referred to earlier} to demonstrate a HV
construction that is not ruled out by his contextuality
no-go theorem. However, his model was contextual (and
\emph{functionally-consistent}) which is in contrast with our
construction which is non-contextual (and has \emph{weak
functional-consistency}).

We end with two short remarks.} First, we illustrate
how \emph{non functional-consistency} gives rise to situations which
could get confused with the presence of contextuality.
Imagine for two spin 1/2 particles
\begin{eqnarray*}
&&\hat{B}_{1}\!=\!\hat{\sigma}_{z}\otimes\hat{\mathbb{I}}\!=\\
&&\!\left|00\right\rangle
\left\langle 00\right|+\left|01\right\rangle
\left\langle 01\right|-\left[\left|10\right\rangle
\left\langle 10\right|+\left|11\right\rangle
\left\langle 11\right|\right], \\
&&\hat{B}_{2}\!=\!\hat{\mathbb{I}}\otimes\hat{\sigma}_{z}\!=\\
&&\!\left|10\right\rangle
\left\langle 10\right|+\left|11\right\rangle
\left\langle 11\right|-\left[\left|00\right\rangle
\left\langle 00\right|+\left|01\right\rangle
\left\langle 01\right|\right], \\
&&\hat{C}=\!f(\{\hat{B}_{i}\}) \\
&&\quad=1.\left|00\right\rangle \left\langle
00\right|+2.\left|01\right\rangle \left\langle
01\right|+ \\
&&\! \quad
3.\left|10\right\rangle \left\langle
10\right|+4.\left|11\right\rangle \left\langle
11\right|. 
\end{eqnarray*}
$\hat{C}$ maybe viewed as a function of
$\hat{B}_{1}$, $\hat{B}_{2}$ and other operators
$\hat{B}_{i}$ which are constructed to obtain a
maximally commuting set. A measurement of
$\hat{C}$, will collapse the state into one of the
states which are simultaneous eigenkets of
$\hat{B}_{1}$ and $\hat{B}_{2}$. Consequently,
from the observed value of $\hat{C}$, one can
deduce the values of $\hat{B}_{1}$ and
$\hat{B}_{2}$. Now consider
$\sqrt{2}\left|\chi\right\rangle
=\left|10\right\rangle +\left|01\right\rangle $,
for which $M_{1}(\hat{B}_{1})=1$, and
$M_{1}(\hat{B}_{2})=1$, using our model, with
$c<0.5$. However, $M_{1}(\hat{C})=1$, from which
one can deduce that $\hat{B}_{1}$ was $+1$, while
$\hat{B}_{2}$ was $-1$.  This property itself, one
may be tempted call contextuality, viz.  the value
of $\hat{B}_{1}$ depends on whether it is measured
alone or with the remaining $\{\hat{B}_{i}\}$.
However, it must be noted that $\hat{B}_{1}$ has a
well defined value, and so does $\hat{C}$. Thus by
our accepted definition, there is no
contextuality. It is just that $M_{1}(\hat{C})\neq
\tilde
f(M_{1}(\hat{B}_{1}),M_{1}(\hat{B}_{2}),\dots)$,
viz. the theory is \emph{non functionally-consistent}. Note that
after measuring $\hat{C}$, however,
$M_{2}(\hat{B}_{1})=+1$ and
$M_{2}(\hat{B}_{2})=-1$ (for any value of $c$)
consistent with those deduced by measuring
$\hat{C}$.  Evidently, \emph{functional-consistency} must hold
for the common eigenkets of $\hat{B}_{i}$'s.
Consequently, any violation of \emph{functional-consistency}
must arise from states that are super-positions or
linear combinations of these eigenkets.

{ Second, note that} entanglement is not necessary to demonstrate
\emph{non functional-consistency}; for instance the
Peres Mermin test is a state independent test
where a separable state can be used to arrive at a
contradiction. On the other hand if everything is
\emph{functionally-consistent} in a situation, can we have
violation of Bell's inequality or non-locality?
The answer is no and, therefore, one can say that
Bell's inequalities bring out  non-local
consequences of \emph{non functionally-consistent} prediction
maps and the notion of \emph{non functional-consistency} is
more basic.

\section{Conclusion}
In this letter we have presented
\emph{non functional-consistency} as an alternative to
contextuality and as an essential signature of
quantumness at the kinematic level. Our result
points to a (quantum) dynamical exploration of contextuality which 
so far has effectively been studied kinematically only.
We expect our result to provide new insights that will be useful
in the areas of foundations of QM and QIP.

\acknowledgments
ASA and Arvind acknowledge
funding from KVPY and DST Grant No.
EMR/2014/000297 respectively.




\bibliographystyle{unsrt}
\bibliography{references}

\end{document}